\documentclass{IEEEtran}
\IEEEoverridecommandlockouts
\usepackage{ifpdf}
\usepackage{cite}
\usepackage{graphicx}
\usepackage{amsthm} 
\usepackage{xcolor}

\ifCLASSINFOpdf
\else
\fi
\usepackage{amsmath}
\usepackage{algorithmic}
\usepackage{array}

\usepackage{stfloats}
\usepackage{url}
\usepackage{cite}
\newtheorem{theorem}{Theorem}
\newtheorem{lemma}{Lemma}
\newcommand{\RomanNumeralCaps}[1]{\MakeUppercase{\romannumeral #1}}

\begin{document}
\title{Analytical Voltage Sensitivity Analysis for Unbalanced Power Distribution System}

\author{\IEEEauthorblockN{Sai Munikoti,~\textit{Student Member, IEEE}, Kumarsinh Jhala,~\textit{Member, IEEE}, Kexing Lai,~\textit{Member, IEEE}, \\ Balasubramaniam Natarajan,~\textit{Senior Member, IEEE}
\thanks{K. Jhala is with the Center for Energy, Environmental, and Economic Systems Analysis in the Energy Systems Division at Argonne National Laboratory (e-mail: kjhala@anl.gov).}
\thanks{S. Munikoti, K. Lai and B. Natarajan are with Electrical and Computer Engineering, Kansas State University, Manhattan, KS-66506, USA, (e-mail: saimunikoti@ksu.edu, klai@ksu.edu, bala@ksu.edu)}}}

\maketitle

\begin{abstract}
Large scale integration of distributed energy resources and electric vehicles in a transactive energy environment present new challenges in terms of voltage stability and fluctuations in a power distribution system. The impact of different level of DER/EV penetration on the voltages across the network is typically quantified through voltage sensitivity analyses. Existing methods of voltage sensitivity analysis are computationally expensive and prior efforts to develop analytical approximation lacks generality and have not been effectively validated. The objective of this work is to provide a new analytical method of voltage sensitivity analysis that has low computational cost and also allows for stochastic analysis of voltage change. This paper first derives an analytical approximation of change in voltage at a particular bus due to change in power consumption at other bus in a radial three phase unbalanced power distribution system. Then, the proposed method is shown to be valid for different load configurations, which demonstrates its generality. The results from our analytical approach is validated via classical load flow simulation of the test system based on IEEE 37 bus network. The proposed method is shown to have good accuracy, and computation complexity is of order $O(1)$, compared to $O(n^{3})$ in classical sensitivity analysis approaches.   
\end{abstract}

% \fi
\IEEEpeerreviewmaketitle

\vspace{-2mm}
\section{Introduction}
\IEEEPARstart{T}{he} integration of distributed energy resources (DERs), electric vehicles (EVs) along with active consumers are some of the key ingredients of future smart grid. By 2050, USA and China plan to meet $80\%$ and $60\%$ of their respective demand through renewables \cite{yang2016china}. However, this integration brings new challenges in terms of voltage stability and fluctuations due to increase in the underlying uncertainty and complexity of the system \cite{sun2019review}. To quantify the impact of different level of DER/EV penetration on voltages across the network, sensitivity analysis is a key enabling tool. Voltage sensitivity analysis (VSA) studies change in voltage at a certain bus as a function of change in complex power at some arbitrary bus in the distribution network. Conventional methods of VSA include the  Newton-Raphson load flow and the perturb-and-observe method. These methods suffer from high computational cost and lack of insights about system states. Within the numerical paradigm, Monte-Carlo simulation based scenario analysis using load flow solutions is used to study the impact of large scale implementation of DERs on voltages. As the size of the distribution system and the DER penetration level increases, number of scenarios as well as the computational complexity of the load flow solution increases. On the other hand, existing analytical approaches to sensitivity analysis are incomplete in terms of validation and scalability. The challenges associated with the distribution system like unbalanced nature of loads, different load configuration (e.g., star and delta) as well as different load types, (e.g., constant power, constant current, constant impedance) have limited the prior analytical efforts in VSA. 
%\textcolor{orange}{Also, Monte-Carlo simulation based scenario analysis using power flow solutions is used as the primary tool to understand the impact of large scale implementation of DERs. As the size of the distribution system and the DER penetration level increases, number of scenarios as well as computational complexity of the power flow solution increases. Hence, computational complexity of such scenario analysis increases exponentially with the size of the network.
Therefore, this work proposes an analytical and computationally efficient method to analyze large scale impact of DERs in the power distribution system.
%Therefore, an computationally efficient analytical approach for voltage sensitivity analysis needs to be developed which is computationally efficient.
This would pave the way for a more general stochastic framework of VSA that can systematically account for the spatio-temporal uncertainties associated with DERs. \cite{jhala2017probabilistic}. 

%\textcolor{orange}{Due to DERs, the load change at actor nodes is random, which introduces uncertainty in voltage change. The proposed framework could be used to estimate the probability of voltage violations and find dominant nodes, which have maximum influence on the voltage sensitivity of critical nodes such as hospitals, schools, etc,. Moreover, the power at the dominant nodes could be controlled to mitigate the voltage violation at the critical nodes}.\\
%In this work, we propose an analytical approximation of voltage sensitivity for an unbalanced distribution system which can be applied to different load configuration and load types.      
\textbf{Related work:}
Voltage sensitivity analysis has been used for voltage control in systems with distributed generation (DG). In these works \cite{aghatehrani2012sensitivity,aghatehrani2012reactive,
weckx2014voltage, valverde2013model}, sensitivity analysis is mainly executed via traditional approaches, which are computationally expensive. For instance in \cite{aghatehrani2012reactive}, Newton-Raphson method is used to control the voltage variations in a PV system. There are very few works which attempt to pursue an analytical approach to sensitivity analysis but most of these approaches are not validated through simulation \cite{brenna2010voltage,zad2016centralized,zad2015optimal}. In \cite{brenna2010voltage}, a new sensitivity method is proposed for reactive power change, and is used to select the most effective generator for controlling the voltage in a system with DGs. However, it is assumed that the power losses are negligible which is not realistic. In \cite{zad2016centralized,zad2015optimal}, active and reactive power of DGs are controlled to keep voltage within safe limits. Sensitivity analysis proposed in these papers characterizes the voltage variation at any bus with respect to voltage at a reference bus, but the theoretical results are not validated with simulations. In our prior work \cite{jhala2017probabilistic}, an analytical bound has been derived for voltage change and the results are validated with a standard test system. However, this is done for balanced single phase system. Building on \cite{jhala2017probabilistic}, \cite{jhala2019dominant} 
identifies the most influential nodes which affect the voltage of critical nodes and this can help in quick restoration of voltage services in case of natural disaster or cyber attacks. Authors in \cite{klonari2016application}, have taken a probabilistic  approach where smart meter measurements are used along with sensitivity analysis to define boundary values of various operation indices. Here, real and reactive power consumption of homes are assumed to be independent which is not the case in reality. Similarly, authors in \cite{valverde2018estimation} have used smart meter data, where 
%hierarchical clustering is done before training a regression model. 
the regression model is used to predict the voltage change. Both \cite{klonari2016application },\cite{valverde2018estimation} are not scalable to large distribution systems and are dependent on data availability. To summarize, existing analytical approaches lack generality and computational efficiency, limiting their application to large scale unbalanced distribution system.\\
\textbf{Contributions:}     
This work proposes an analytical approximation of voltage sensitivity in a general three phase distribution system. Major contributions of this paper includes :
\begin{itemize}
\item An analytical approximation of voltage change in any bus due to change in power at any arbitrary bus in a unbalanced radial distribution network is derived for the first time in Section \RomanNumeralCaps{2}. This approximation is validated by simulation of IEEE 37 bus system in Section \RomanNumeralCaps{3}.
\item The approximation holds true for star and delta configuration with different impedance matrices for each case. 
\item The computational complexity of proposed method is  $O(1)$, compared to $O(n^{3})$ in classical NR method.
%\item The approximation also works for constant power, constant current and constant impedance loads.
%\item The computational complexity of the proposed algorithm is significantly lower, compared to the traditional sensitivity analysis methods that employ load flow calculations.
\end{itemize} 
%In the next section, a brief review of traditional VSA methods has been done followed by the derivation of the proposed analytical approximation.  
%The rest of the paper is organized as follows: Section \RomanNumeralCaps{2} discusses the traditional VSA methods followed by the derivation of the proposed analytical approximation. Section \RomanNumeralCaps{3} presents the simulation results. The conclusion and future works are presented in Section \RomanNumeralCaps{4}. 
\section{Voltage Sensitivity Analysis}
\subsection{Traditional voltage sensitivity analysis}
VSA implies the sensitivity of voltage (voltage change) with respect to change in complex power. A sensitivity matrix is normally developed to visualize the coupling between change of voltage magnitude and angle, and change of power consumption/injection.
%consider a distribution network with $N$ nodes with $P_{i}, Q_{i}$ representing the real and reactive power at node $i$respectively. If there is a change in power at any node, there is a change in both the magnitude $|\Delta V_{j}|$ and angle of voltage $\delta_{j}$ at any arbitrary node j. Sensitive matrix relates this change in voltage magnitude and angle with the change in power consumption or injection.
Two Traditional approaches of sensitivity analysis are Newton-Raphson (NR) and Perturb-and-Observe method, but NR method is popular. To understand NR method, let us consider a distribution network with $N$ buses where  $P_{i}, Q_{i}$ represent the real and reactive power at some node $i$ respectively. Let $|\Delta V_{j}|$, $\delta_{j}$ represent the changes in voltage magnitude and angle triggered by the power change at a certain node. In the iterative NR method, a set of non-linear equations is used to  relate voltage changes($|\Delta V_{j}|$, $\delta_{j}$) with power change($P_{i}, Q_{i}$), which is approximated by linear operator, popularly known as Jacobian/Sensitivity matrix. 
%as shown in: 
%\begin{equation}
%\begin{bmatrix}
%\Delta \delta_{2}(i)\\
%\vdots\\
%\Delta \delta_{N}(i) \\
%\hline
%\Delta V_{2}(i) \\
%\vdots\\
%\Delta V_{N}(i)
%\end{bmatrix} = J^{-1} \begin{bmatrix}
%\Delta P_{2}(i)\\
%\vdots\\
%\Delta P_{N}(i) \\
%\hline
%\Delta Q_{2}(i) \\
%\vdots\\
%\Delta Q_{N}(i)
%\end{bmatrix}
%\end{equation}
%where J = $\begin{bmatrix}
%\frac{\partial P_{2}}{\partial \delta_{2} } \hdots \frac{\partial P_{2}}{\partial \delta_{N} } \vline  \frac{\partial P_{2}}{\partial V_{2} } \hdots \frac{\partial P_{2}}{\partial V_{N} } \\
%\vdots  \\
%\frac{\partial P_{N}}{\partial \delta_{2} } \hdots \frac{\partial P_{N}}{\partial \delta_{N} } \vline \frac{\partial P_{N}}{\partial V_{2} } \hdots \frac{\partial P_{N}}{\partial V_{N} }\\
%\hline \\
%\frac{\partial Q_{2}}{\partial \delta_{2} } \hdots \frac{\partial Q_{2}}{\partial \delta_{N} } \vline \frac{\partial Q_{2}}{\partial V_{2} } \hdots \frac{\partial Q_{2}}{\partial V_{N} } \\
%\vdots \\
%\frac{\partial Q_{N}}{\partial \delta_{2} } \hdots \frac{\partial Q_{N}}{\partial \delta_{N} } \vline \frac{\partial Q_{N}}{\partial V_{2} } \hdots \frac{\partial Q_{N}}{\partial V_{N} } \\
%\end{bmatrix}$ \\ \\
%$J$ is the sensitivity/Jacobian matrix.
%Here bus one is not included in calculation because it is generally considered as slack bus with stable voltage of $1\angle 0$ p.u.
With the initial setting of all the variables, the voltage is updated iteratively until the convergence criterion is met. This method is computationally expensive and the complexity grows with the size of the network as $O(n^3)$. Also, the Jacobian matrix is valid only for a specific state of the system and needs to be recomputed, if there are major changes in the states of the system.
%An alternate sensitivity analysis method is the  perturb-and-observe approach where a small perturbation in complex bus power is done to measure the impact on the bus voltages. This approach is also computationally complex as it requires the entire state to be recomputed for any changes in the network.
%There are some demerits of these approaches. Perturb and observe method is not much efficient because it recomputes the entire state for any changes in the network.
%NR method is computationally expensive when size of the network is very large which is true in case of real scenarios. Also, Jacobian matrix is valid only for specific state of the system and need to recompute for major changes in the states of the system.
Moreover, this method is purely numeric, i.e., the computed sensitivity matrix does not provide any analytical insights. Thus, any stochastic analysis done using this method, is valid for a specific state of the system. Therefore, there is a need for more general and computationally efficient approach which can be used for future distribution grid analysis. 
%For the future prospective of the grid, when states of the system are random, for instance node voltages are random in case of PV injection, we need analytical method of sensitive analysis to model the probability distribution of voltage change. 
%Next section provides an computationally efficient approximation of voltage change in three phase distribution system.    
   
\subsection{Analytical approximation of voltage sensitivity analysis}
In this subsection, we propose an analytical method for VSA in a three phase distribution system. Change in power at any one phase of a bus results in voltage change at all buses of the distribution system. Nodes where power is changing are referred to as actor nodes, whereas the node where voltage change is monitored is termed as an observation node. For now, assume that there is a single actor node, i.e., power is changing at a single node, the load model is a constant power load with star configuration and source bus is a slack bus. In our previous work, we have proposed a 1about VSA for balanced radial distribution network, which can be found in Theorem 1 of \cite{jhala2017probabilistic}. We further extend this for three phase unbalanced distribution system and develop  Theorem 1 as :
%\vspace{-2mm}
%\begin{figure}[h!]
%	\centering
%	\includegraphics[width = 7cm,height=3cm]{Network_eg.png}
%	\caption{Example network with $\theta$ as actor node and $O$ as observation node }
%	\label{fig:1}
%\end{figure}
\vspace{-2.0mm}
\begin{theorem}
For a unbalanced power distribution system, change in complex voltage $\Delta V_{O\theta}$ at an observation node due to change in complex power of an actor node can be approximated by 
\begin{equation}	
\begin{bmatrix}
\Delta V_{O\theta}^{a} \\
\Delta V_{O\theta}^{b} \\
\Delta V_{O\theta}^{c} 
\end{bmatrix} \approx -	
\begin{bmatrix}
\frac{\Delta S_{\theta a}^{\star}Z_{aa}}{V_{\theta a}^{\star }} + \frac{\Delta S_{\theta b}^{\star}Z_{ab}}{V_{\theta b}^{\star }}+ \frac{\Delta S_{\theta c}^{\star}Z_{ac}}{V_{\theta c}^{\star }}  \\
\frac{\Delta S_{\theta a}^{\star}Z_{ba}}{V_{\theta a}^{\star }} + \frac{\Delta S_{\theta b}^{\star}Z_{bb}}{V_{\theta b}^{\star }}+ \frac{\Delta S_{\theta c}^{\star}Z_{bc}}{V_{\theta c}^{\star }}  \\
\frac{\Delta S_{\theta a}^{\star}Z_{ca}}{V_{\theta a}^{\star }} + \frac{\Delta S_{\theta b}^{\star}Z_{cb}}{V_{\theta b}^{\star }}+ \frac{\Delta S_{\theta c}^{\star}Z_{cc}}{V_{\theta c}^{\star }} 
\end{bmatrix}
\label{eq:2} 
\end{equation}
where $a,b$ and $c$ represents the three phases, and this notation is used throughout the paper; $V_{\theta}^{*}$ and $\Delta S_{\theta}$ represent complex conjugate of voltage and complex power change at actor node $\theta$, respectively; $Z$ denotes the self or mutual impedance of the shared path between observation node and actor node from source node.
\end{theorem}
\begin{proof} Voltage at an observation node can be computed in terms of the difference between voltage at the source node and sum of the voltage drops across all lines/edges between the source node and observation node. Let $E_{o}$ be set of all edges between the source node and the observation node. Using KVL, voltage at observation node $o$ can be written as:
\begin{equation}
\begin{bmatrix}
V_{Oa} \\[3pt]
V_{Ob} \\[3pt]
V_{Oc} 
\end{bmatrix} =\begin{bmatrix}
V_{Sa} \\[3pt]
V_{Sb} \\[3pt]
V_{Sc} 
\end{bmatrix} -\sum_{e\epsilon E_{o}}\begin{bmatrix}
V_{ea}^{d} \\[3pt]
V_{eb}^{d} \\[3pt]
V_{ec}^{d}  
\end{bmatrix}
\label{eq:3} 	
\end{equation}
where $V_{O}$, $V_{S}$, and $V_{e}^{d}$ are voltage at observation node, voltage of source node, and the voltage drop across edge $e$, respectively. Let $I_{e}$ and $Z_{e}$ be the current and impedance for edge $e$. Here, along with self impedance, mutual impedance of the line will also contribute to the voltage drop. In LV distribution network, value of shunt impedance can be ignored. We can represent (\ref{eq:3}) in a form incorporating line current and impedance, denoted by $I_{e}$ and $Z_{e}$ as:
\begin{equation}
\pmb{V_{O}} = \pmb{V_{S}} - \sum_{e\epsilon E_{o}} I_{e}\pmb{Z_{e}}
\label{eq:4}
\end{equation}
where $\pmb{V_{O}}=\begin{bmatrix}
V_{Oa} \\
V_{Ob} \\
V_{Oc} 
\end{bmatrix}$ and  $\pmb{Z_{e}}=\begin{bmatrix}
z_{e,aa} &z_{e,ab} & z_{e,ac} \\
z_{e,ba} &z_{e,bb} & z_{e,bc}  \\
z_{e,ca} &z_{e,cb} & z_{e,cc} 
\end{bmatrix}$ \\ \\ 
Let $S_{n}$ be complex power consumption or injection at node $n$ and $V_{n}^{*}$ be the complex conjugate of voltage at node n. The current flowing through a particular phase of edge $e$ can be written as $\sum_{n\epsilon N_{e}}\frac{S_{n}}{V_{n}^{*}}$ where $N_{e}$ is the set of all nodes $n$ for which edge $e$ is between node $n$ and source node. Power from the source node to all the nodes in the set $N_{e}$ flows through edge $e$. Therefore, current in edge $e$ will be affected by the power change at nodes $ n \epsilon N_{e}$. Therefore, the voltage at the observation node can be written as: 
\begin{equation}
	\pmb{ V_{O}} =\pmb{V_{S}} -\sum_{e\epsilon E_{o}}  \sum_{n\epsilon N_{e}} \pmb{Z_{e}} 
	\begin{bmatrix}
	\frac{S_{na}^{\star}}{V_{na}^{\star}} \hspace{0.1cm} \frac{S_{nb}^{\star}}{V_{nb}^{\star}} \hspace{0.1cm} \frac{S_{nc}^{\star}}{V_{nc}^{\star}}
	\end{bmatrix}^{T}
	\label{eq:5} 
\end{equation}  
When power consumption of node $n$ changes from $S_{n}$ to $S_{n}^{'}$, the voltage will change from $V_{n}$ to $V_{n}^{'}$ and consequently voltage at observation node will change to $V_{o}^{'}$. The new voltage at observation node can be written as:
\begin{equation}
\pmb{V_{O}^{'}} = \pmb{V_{S}} - \sum_{e\epsilon E_{o}}  \sum_{n\epsilon N_{e}} \pmb{Z_{e}} 
\begin{bmatrix}
\frac{S_{na}^{'\star}}{V_{na}^{'\star}} \hspace{0.1cm}
\frac{S_{nb}^{'\star}}{V_{nb}^{'\star}} \hspace{0.1cm}
\frac{S_{nc}^{'\star}}{V_{nc}^{'\star}}
\end{bmatrix}^{T}
\label{eq:6} 
\end{equation} 
where  $S_{n}^{'\star}= S_{n}^{\star}+\Delta S_{n}^{\star}$ and $V_{n}^{'}= V_{n}+\Delta V_{n}$. The effective voltage change at observation node can be written as $\Delta V_{o}=V_{o}-V_{o}^{'}$. Using (\ref{eq:5}) and (\ref{eq:6}), change in voltage at observation node can be expressed as:
\begin{equation}
\begin{split}
\pmb{\Delta V_{O}} 
%& =\sum_{e\epsilon E_{o}}  \sum_{n\epsilon N_{e}} \pmb{Z_{e}} 
%\begin{bmatrix}
%\frac{S_{na}}{V_{na}^{\star}} \hspace{0.1cm}
%\frac{S_{nb}}{V_{nb}^{\star}} \hspace{0.1cm}
%\frac{S_{nc}}{V_{nc}^{\star}}
%\end{bmatrix}^{T} -\\
%& \hspace{2.5cm}\sum_{e\epsilon E_{o}}  \sum_{n\epsilon N_{e}} \pmb{Z_{e}} 
%\begin{bmatrix}
%\frac{S_{na}^{'}}{V_{na}^{'\star}} \hspace{0.1cm}
%\frac{S_{nb}^{'}}{V_{nb}^{'\star}} \hspace{0.1cm}
%\frac{S_{nc}^{'}}{V_{nc}^{'\star}} 
%\end{bmatrix}^{T} \\
& =\sum_{e\epsilon E_{o}} \left(\sum_{n\epsilon N_{e}} \begin{bmatrix}
\frac{S_{na}^{\star}}{V_{na}^{\star}} - \frac{S_{na}^{\star}+ \Delta S_{na}^{\star}}{V_{na}+\Delta V_{na}} \\[5pt]
\frac{S_{nb}^{\star}}{V_{nb}^{\star}} - \frac{S_{nb}^{\star}+ \Delta S_{nb}^{\star}}{V_{nb}+\Delta V_{nb}}\\[5pt]
\frac{S_{nc}^{\star}}{V_{nc}^{\star}} - \frac{S_{nc}^{\star} + \Delta S_{nc}^{\star}}{V_{nc}+\Delta V_{nc}}
\end{bmatrix} \right) \pmb{Z_{e}} \\
& =\sum_{e\epsilon E_{o}} \left( \sum_{n\epsilon N_{e}} \begin{bmatrix}
\frac{S_{na}^{\star}\Delta V_{na}^{\star}-\Delta S_{na}^{\star} V_{na}}{V_{na}^{\star}(V_{na}^{\star}+\Delta V_{na}^{*})}  \\[5pt]
\frac{S_{nb}^{\star}\Delta V_{nb}^{\star}-\Delta S_{nb}^{\star} V_{nb}}{V_{nb}^{\star}(V_{nb}^{\star}+\Delta V_{nb}^{*})}  \\[5pt]
\frac{S_{nc}^{\star}\Delta V_{nc}^{\star}-\Delta S_{nc}^{\star} V_{nc}}{V_{nc}^{\star}(V_{nc}^{\star}+\Delta V_{nc}^{*})} 
\end{bmatrix} \right) \pmb{Z_{e}}
\end{split}
\label{eq:7} 
\end{equation}
In practice, voltage changes are typically small compared to actual node voltage. Hence, it is reasonable to assume that $\Delta V_{n}^{*}/(V_{n}^{*}+\Delta V_{n}^{*}) \to 0 $. Thus, (\ref{eq:7}) can be approximated as:
\begin{equation}
\pmb{\Delta V_{O}}
=\sum_{e\epsilon E_{o}} \left( \sum_{n\epsilon N_{e}} \begin{bmatrix}
\frac{-\Delta S_{na} ^{\star}}{V_{na}^{\star}+\Delta V_{na}^{*}}  \\[5pt]
\frac{-\Delta S_{nb} ^{\star}}{V_{nb}^{\star}+\Delta V_{nb}^{*}}  \\[5pt]
\frac{-\Delta S_{nc} ^{\star}}{V_{nc}^{\star}+\Delta V_{nc}^{*}} 
\end{bmatrix} \right) \pmb{Z_{e}}
\label{eq:8} 
\end{equation}  
Equation (\ref{eq:8}) holds true even for the cases when there are multiple actor nodes. However, here only one actor node is changing its power consumption, therefore, $\Delta S_{n}$ is zero for all nodes except actor node $\theta$. Let $E_{\theta }$ be set of all edges between the actor node and source node. When actor node $\theta$ changes power consumption, current flowing through the edges changes for all edges of set $E_{\theta}$. Voltage drop across the edges between the source node and observation node, changes only for edges that belongs to subset $E_{\theta } \cap E_{o} $. 
\begin{equation} 
\begin{split}
\pmb{\Delta V_{O\theta}} 
%& =\sum_{e\epsilon E_{o}\cap E_{\theta}} \left( \begin{bmatrix}
%\frac{-\Delta S_{\theta a} }{V_{\theta a}^{\star}+\Delta V_{\theta a}^{*}}  \\[5pt]
%\frac{-\Delta S_{\theta b} }{V_{\theta b}^{\star}+\Delta V_{\theta b}^{*}}  \\[5pt]
%\frac{-\Delta S_{\theta c} }{V_{\theta c}^{\star}+\Delta V_{\theta c}^{*}} 
%\end{bmatrix} \right) \pmb{Z_{e }} \\
%& 
=\left( \begin{bmatrix}
\frac{-\Delta S_{\theta a}^{\star} }{V_{\theta a}^{\star}+\Delta V_{\theta a}^{*}}  \\[5pt]
\frac{-\Delta S_{\theta b}^{\star} }{V_{\theta b}^{\star}+\Delta V_{\theta b}^{*}}  \\[5pt]
\frac{-\Delta S_{\theta c}^{\star} }{V_{\theta c}^{\star}+\Delta V_{\theta c}^{*}} 
\end{bmatrix} \right) \pmb{Z_{ O \theta }}
\end{split}
\label{eq:9} 
\end{equation} 
where $\pmb{Z_{ O \theta }} = \sum_{e\epsilon E_{o}\cap E_{\theta}} \pmb{Z_{e}}$ is the impedance matrix. Here, each component in the summation is the impedance of the shared path between the actor node and observation node from source node. We decompose (\ref{eq:9}) into real and imaginary components as follows:
\begin{equation}
\begin{split}
\pmb{\Delta V_{O\theta}^{r}} = -
\begin{bmatrix}
\frac{(\Delta P_{\theta a}R_{o\theta,aa}+\Delta Q_{\theta a}X_{o\theta,aa})(V_{\theta a}^{\star r}+\Delta V_{\theta a}^{\star r})}{(V_{\theta a}^{\star r}+ \Delta V_{\theta a}^{\star r})^{2}+(V_{\theta a}^{\star i}+ \Delta V_{\theta a}^{\star i})^{2}} - \\ 
\frac{(\Delta P_{\theta a}X_{o\theta,aa}-\Delta Q_{\theta a}R_{o\theta,aa})(V_{\theta a}^{\star i}+\Delta V_{\theta a}^{\star i})}{(V_{\theta a}^{\star r}+ \Delta V_{\theta a}^{\star r})^{2}+(V_{\theta a}^{\star i}+ \Delta V_{\theta a}^{\star i})^{2}}+ \dots \\[6pt]
\vdots 
\end{bmatrix}_{3\times 1} \\
\pmb{\Delta V_{O\theta}^{i}} = -
\begin{bmatrix}
\frac{(\Delta P_{\theta a}X_{o\theta,aa}-\Delta Q_{\theta a}R_{o\theta,aa})(V_{\theta a}^{\star r}+\Delta V_{\theta a}^{\star r})}{(V_{\theta a}^{\star r}+ \Delta V_{\theta a}^{\star r})^{2}+(V_{\theta a}^{\star i}+ \Delta V_{\theta a}^{\star i})^{2}} + \\ 
\frac{(\Delta P_{\theta a}R_{o\theta,aa} + \Delta Q_{\theta a}X_{o\theta,aa})(V_{\theta a}^{\star i}+\Delta V_{\theta a}^{\star i})}{(V_{\theta a}^{\star r}+ \Delta V_{\theta a}^{\star r})^{2}+(V_{\theta a}^{\star i}+ \Delta V_{\theta a}^{\star i})^{2}}+ \dots \\[6pt]
\vdots 
\end{bmatrix}_{3\times 1}
\end{split}
\label{eq:10} 
\end{equation}
where superscript $r$ and $i$ represents the real and imaginary components, respectively. In a distribution network, voltage angle relative to source node and magnitude of voltage change is usually very small.
%making $V_{\theta}^{i}$ and $\Delta V_{\theta}^{i}$ very small.
Under the above assumptions, the real part of voltage change can be approximated as:
\begin{equation}
\begin{split}
\pmb{\Delta V_{O\theta}^{r}} \approx -
\begin{bmatrix}
\frac{(\Delta P_{\theta a}R_{o\theta,aa}+\Delta Q_{\theta a}X_{o\theta,aa})(V_{\theta a}^{\star r})}{(V_{\theta a}^{\star r})^{2}+(V_{\theta a}^{\star i})^{2}} - \\ 
\frac{(\Delta P_{\theta a}X_{o\theta,aa}-\Delta Q_{\theta a}R_{o\theta,aa})(V_{\theta a}^{\star i})}{(V_{\theta a}^{\star r})^{2}+(V_{\theta a}^{\star i})^{2}}+ \dots \\[6pt]
\vdots 
\end{bmatrix}_{3\times 1}
\end{split}
\label{eq:11} 
\end{equation}
Similarly, with the same arguments, the imaginary part can be approximated.
%\begin{equation}
%\begin{split}
%\pmb{\Delta V_{O\theta}^{i}} \approx -
%\begin{bmatrix}
%\frac{(\Delta Q_{\theta a}R_{o\theta,aa}+\Delta P_{\theta a}X_{o\theta,aa})(V_{\theta a}^{\star r})}{(V_{\theta a}^{\star r})^{2}+(V_{\theta a}^{\star i})^{2}} - \\ 
%\frac{(\Delta P_{\theta a}R_{o\theta,aa}-\Delta Q_{\theta a}X_{o\theta,aa})(V_{\theta a}^{\star i})}{(V_{\theta a}^{\star r})^{2}+(V_{\theta a}^{\star i})^{2}}+ \dots \\[6pt]
%\vdots 
%\end{bmatrix}_{3\times 1}
%\end{split}
%\label{eq:12} 
%\end{equation}
By recombining the real and imaginary parts, the approximate voltage change for all three phases can be written as:
\begin{equation}	
\pmb{\Delta V_{O\theta}} \approx-	
\begin{bmatrix}
\frac{\Delta S_{\theta a}^{\star}Z_{aa}}{V_{\theta a}^{\star }} + \frac{\Delta S_{\theta b}^{\star}Z_{ab}}{V_{\theta b}^{\star }}+ \frac{\Delta S_{\theta c}^{\star}Z_{ac}}{V_{\theta c}^{\star }}  \\
\frac{\Delta S_{\theta a}^{\star}Z_{ba}}{V_{\theta a}^{\star }} + \frac{\Delta S_{\theta b}^{\star}Z_{bb}}{V_{\theta b}^{\star }}+ \frac{\Delta S_{\theta c}^{\star}Z_{bc}}{V_{\theta c}^{\star }}  \\
\frac{\Delta S_{\theta a}^{\star}Z_{ca}}{V_{\theta a}^{\star }} + \frac{\Delta S_{\theta b}^{\star}Z_{cb}}{V_{\theta b}^{\star }}+ \frac{\Delta S_{\theta c}^{\star}Z_{cc}}{V_{\theta c}^{\star }} 
\end{bmatrix}
\label{eq:13} 
\end{equation}  
From equation (\ref{eq:13}), we can see that voltage change at an observation node depends on the power change of the actor node and their voltage. Besides, location of the actor and observation nodes will also affect the result, as impedance matrix $Z_{O\theta}$ relies on the location.
\end{proof}
\begin{lemma}
\textbf{Multiple actor nodes}: If there are multiple actor nodes in the unbalanced distribution network, net change in complex voltage at the observation node $O$ due to aggregate effect of all actor nodes is approximated as:
\begin{equation}	
\begin{bmatrix}
\Delta V_{O}^{a} \\
\Delta V_{O}^{b} \\
\Delta V_{O}^{c} 
\end{bmatrix} \approx- \sum_{\theta \epsilon A}  \left(	
\begin{bmatrix}
\frac{\Delta S_{\theta a}^{\star}Z_{aa}}{V_{\theta a}^{\star }} + \frac{\Delta S_{\theta b}^{\star}Z_{ab}}{V_{\theta b}^{\star }}+ \frac{\Delta S_{\theta c}^{\star}Z_{ac}}{V_{\theta c}^{\star }} \\
\frac{\Delta S_{\theta a}^{\star}Z_{ba}}{V_{\theta a}^{\star }} + \frac{\Delta S_{\theta b}^{\star}Z_{bb}}{V_{\theta b}^{\star }}+ \frac{\Delta S_{\theta c}^{\star}Z_{bc}}{V_{\theta c}^{\star }}  \\
\frac{\Delta S_{\theta a}^{\star}Z_{ca}}{V_{\theta a}^{\star }} + \frac{\Delta S_{\theta b}^{\star}Z_{cb}}{V_{\theta b}^{\star }}+ \frac{\Delta S_{\theta c}^{\star}Z_{cc}}{V_{\theta c}^{\star }}
\end{bmatrix}\right) 
\label{eq:22} 
\end{equation} 
where $A$ is the set of all actor nodes.
\end{lemma}
\begin{proof}
The results follows the same procedure as in Lemma 1 of \cite{jhala2017probabilistic} and has been omitted due to page constraint.
\end{proof}
\subsection{Analytical Approximation for star and delta connection}
Star (wye) and Delta are the two most common load configurations. The Analytical approximation in (\ref{eq:13}) is derived for a star load configuration. 
%In Star connection line current is equal to phase current and hence the line current $I_{e}$ in equation $(4)$ is expressed directly in terms of phase power and voltage i.e $I_{e} = \sum_{n\epsilon N_{e}} (S_{n}/V_{n})$.
In delta, line current is not equal to phase current, rather it is $\sqrt{3}$ times the phase current, which is not true for an unbalanced condition. So, line current for delta load is expressed as the difference of phase currents as shown in equation (\ref{eq:14}), which is valid for both balanced and unbalanced conditions, i.e.,
\begin{equation}
\begin{split}
I_{A} = I_{a}-I_{c} = \frac{S_{a}^{\star}}{V_{a}^{\star}} - \frac{S_{c}^{\star}}{V_{c}^{\star}} \\
%I_{B} = I_{b}-I_{a} = \frac{S_{b}}{V_{b}} - \frac{S_{a}}{V_{a}} \\
%I_{C} = I_{c}-I_{b} = \frac{S_{c}}{V_{c}} - \frac{S_{b}}{V_{b}}
\end{split}
\label{eq:14} 
\end{equation}
Here $I_{A}$ denotes the line current, and $I_{a}, I_{c}$ represents the phase currents. Expressing line currents for all phases in the form of (\ref{eq:14}), voltage of phase $a$ in $(\ref{eq:4})$ is reformulated as:   
\begin{equation}
\begin{split}
{\scriptstyle
V_{oa} = 
%V_{sa} - [Z_{aa}(\frac{S_{a}}{V_{a}} - \frac{S_{c}}{V_{c}}) +Z_{ab}(\frac{S_{b}}{V_{b}} - \frac{S_{a}}{V_{a}}) + \\ Z_{ac}(\frac{S_{c}}{V_{c}} - \frac{S_{b}}{V_{b}})] \\
V_{sa} - [\frac{S_{a}^{\star}}{V_{a}^{\star}}(Z_{aa}-Z_{ab}) +\frac{S_{b}^{\star}}{V_{b}^{\star}}(Z_{ab}-Z_{ac}) + \frac{S_{c}^{\star}}{V_{c}^{\star}}(Z_{ac}-Z_{aa})]}
\end{split}
\label{eq:15} 
\end{equation}
%Now putting the line current for all three phases gives the following expression:
Now, we can express (\ref{eq:15}) for all the three phases in a compressed form as:
\begin{equation}
\pmb{V_{O}} = \pmb{V_{S}} - \sum_{e\epsilon E_{o}} I_{e}\pmb{Z_{e}}
\label{eq:16}
\end{equation}
where $ \pmb{Z_{e}} = \begin{bmatrix} 
	z_{e,aa}-z_{e,ab} &z_{e,ab}-z_{e,ac} & z_{e,ac}-z_{e,aa} \\
	z_{e,ba}-z_{e,bb} &z_{e,bb}-z_{e,bc} & z_{e,bc}-z_{e,ba}  \\
	z_{e,ca}-z_{e,cb} &z_{e,cb}-z_{e,cc} & z_{e,cc}-z_{e,ca} 
\end{bmatrix} $ \\

Here, we can observe that (\ref{eq:16}) is similar to (\ref{eq:4}), except that the impedance matrix $\pmb{Z_{e}}$, whose each entry now involves the binary component unlike a single component in the star connection case. Thus, the derived analytical expression  (\ref{eq:13}) is valid for both the load configurations, with a slightly different impedance matrix for each case. 
\vspace{-0.5mm}
\section{Results}
This section verifies the analytical method of VSA based on the simulation of the modified IEEE 37 bus system shown in Fig.\ref{fig:2}. This test bed is selected due to its highly unbalanced nature and has been used by various researchers \cite{khushalani2007development} in the past for validation.
The nominal voltage of this system is 4.8 kV. Classical load flow method is used as a baseline method for validating our proposed method. 
To evaluate the performance of the proposed method, we consider two scenarios to first examine the accuracy and then analyzes the various advantages. First scenario is designed to study the accuracy of the proposed method across different power variations. Here, nodes $22$ and $9$ are the actor and the observation nodes, respectively with rated power of $42$ kW each. Power of node $22$ is varied from $0$ to $84$ kW in steps of $7$ kW. The maximum variation is nearly $100\%$ of the base load, in both the directions, which is more than the typical range encountered in practice. Positive load change indicates an increase in power consumption or decrease in power injection, whereas negative load change denotes decrease in power consumption or increase in injection. The change in power at any one phase, affects the voltages of all phases. First, we have tested the efficiency of the proposed method for same phase case, i.e., power is varied in phase $c$ of actor node $22$ and voltage change is monitored for the same phase of observation node $9$, which is shown in Fig. \ref{fig:3}. From Fig. \ref{fig:3}, it can be observed that the proposed analytical method approximates the baseline method for a sufficiently large range of power deviation, which demonstrates its accuracy. For approximate deviation upto $80\%$ of the base load, the analytical method almost coincides with simulated value, with small error in the range of $10^{-4}$ to $10^{-3}$ pu for power deviation outside this limit. To test the effectiveness of the proposed method for cross phase effect, voltage change is recorded for phase $a$ of observation node $9$, which is shown in  Fig. \ref{fig:3}. Here also, the proposed method approximates the baseline with a good accuracy over a large range of power deviation. This is because the voltage change factor from each phase is incorporated in our voltage sensitivity approximation. From  Fig. \ref{fig:3}, it can also be inferred that the coupling between power and voltage changes in the same phase is higher (around five times) than in the cross phase. Change in voltage in phase $a$ due to change in power at phase $c$ is caused by mutual inductance between different phases of the line and mutual inductance is less than self inductance.

The second scenario is designed to study the accuracy of the proposed approach across different nodes. Here, power drawn by phase $c$ of node $22$ is increased by $21$ kW and voltage change is observed across all nodes. The voltage change for all the phases is shown in the  Fig. \ref{fig:4}. It can be seen that $\Delta V$ is higher for nodes closer to the actor node as the length of shared conductor between the observation and the actor node is large. On the other hand, $\Delta V$ is small for nodes closer to substation, as the length of shared conductor and impedance of shared line is smaller. It can be observed that the accuracy of the proposed method is very high with more than $95\%$, but it is slightly low in the cross phase of certain nodes, that have relatively high voltage change. 

%As expected, higher power change of actor node will lead to more significant voltage change, which implies potential voltage stability issue with DER and active loads integrations. In addition, our proposed analytical method approximates the baseline method that demonstrates its accuracy.  

\begin{figure}[h!]
	\centering
	\includegraphics[width = 6.9cm,height=3cm]{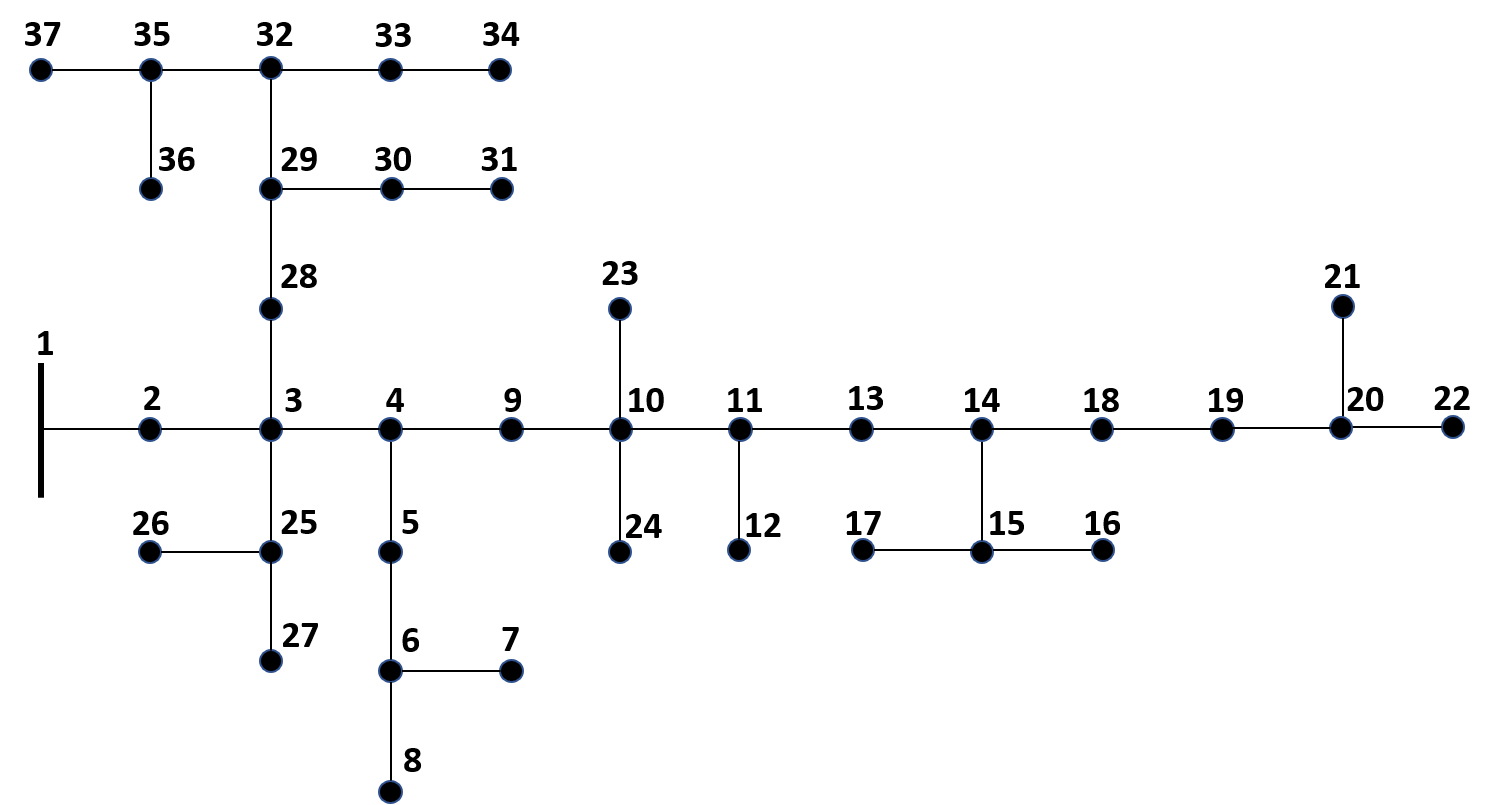}
	\caption{Modified IEEE 37 bus network}
	\label{fig:2}
\end{figure}
\begin{figure}[h!]
	\includegraphics[width = 8.2cm,height=6cm]{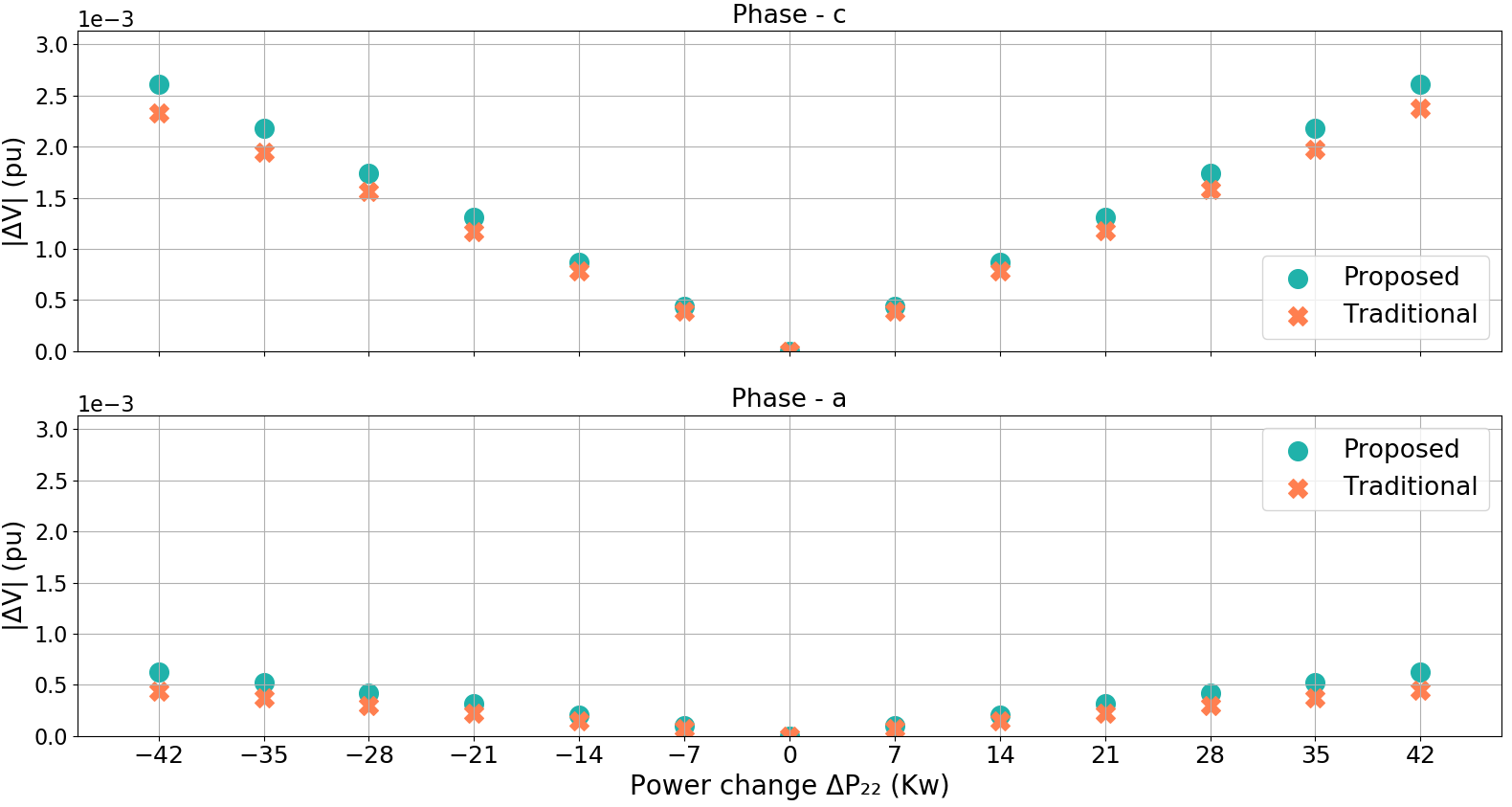}
	\caption{Change in voltage at phase $c$ and $a$ of node $9$, when actor node is $22$}
	\label{fig:3}
\end{figure}
\begin{figure}[h!]
	\includegraphics[width = 8.8cm,height=6cm]{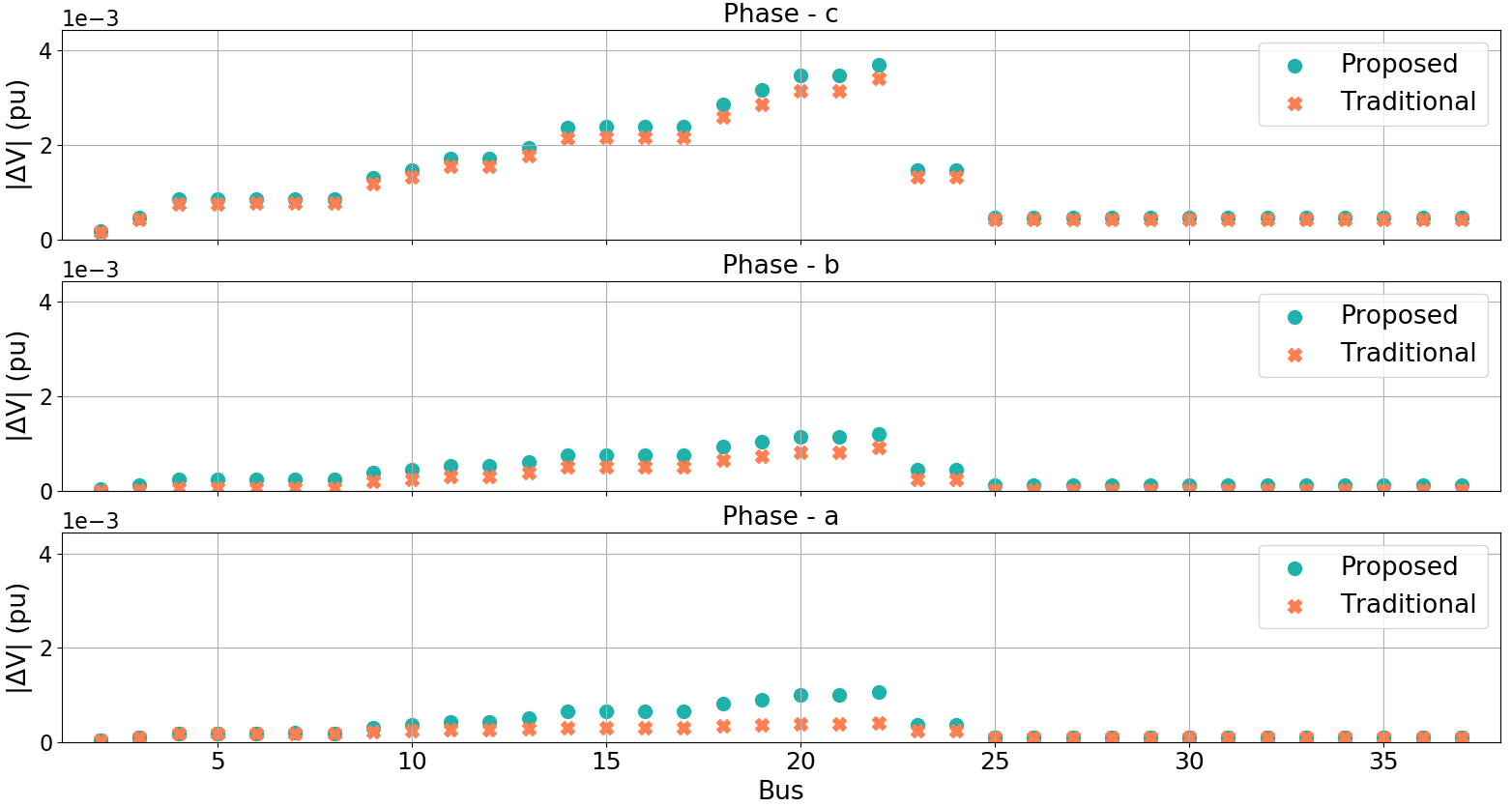}
	\caption{Change in voltage at phase $c,b,a$ of different observation nodes when actor node is $22$}
	\label{fig:4}	
\end{figure}
\vspace{-0.2cm} 

The proposed approach has various advantages over existing methods.
First of all, it addresses the computational shortcoming of numerical approaches.
%It has significantly low computational complexity. 
%Computational complexity is generally measured in two ways, i.e., system execution time and the scaling of number of operations captured via Big O. 
The complexity of the proposed analytical method is of the order $O(1)$, because the calculation of voltage change in (\ref{eq:13}) does not scale with the size of the network ($n$). While, the complexity in classical Newton-Raphson method is of order $O(n^{3})$, as it involves the inversion operation of the Jacobian matrix. Specifically, the execution time of our method to calculate the voltage sensitivity for a single observation node is $0.00871s$, compared to $0.0537s$ in classical load flow method with intel i7 processor. Thus, the proposed method is nearly ten times faster and this factor will increase significantly as the size of the network increases. For example, if we consider IEEE 342 node system, the complexity of the traditional method will be $O(342^3)$ with the execution time growing in a similar fashion. On the other hand, the time taken by our method will be similar to the $37$ bus system as complexity remains constant.
Secondly, it addresses the generalization problem of analytical approaches by appropriately validating the proposed method with standard three phase unbalanced test system. Finally, because of the first two advantages, the proposed framework could also allow us to perform stochastic analysis which is not possible with traditional approaches. Random change in power caused by renewable DERs causes random fluctuations in the distribution system voltage. The proposed approach can further be used to estimate the probability of voltage violations and find dominant nodes which have maximum influence on the voltage sensitivity \cite{jhala2019dominant}. Identification of dominant nodes can later be used to mitigate the voltage violation of specific critical nodes. This clearly shows that the proposed approach has an edge over traditional methods. 
\vspace{-0.5mm}
\section{Conclusion}
The objective of this work is to derive an analytical expression of voltage sensitivity due to power change at certain locations of an unbalanced distribution system. The major improvement of the proposed analytical method over conventional methods, is the lower computational cost. Furthermore, it lay the groundwork for stochastic analysis of large scale penetration of DERs. We have examined the fidelity of the analytical approximation under various load change scenarios, with different load configurations. The results demonstrate the accuracy of the proposed method relative to conventional baseline method.
As part of future work, we plan to derive error bounds associated with our analytical approximation, validate in larger test network, and extend this work to a stochastic framework incorporating spatio-temporal uncertainty.
\section*{Acknowledgment}
This material is based upon work partly supported by the Department of Energy, Office of Energy Efficiency and Renewable Energy (EERE), Solar Energy Technologies Office, under Award \# DE-EE0008767 and National science foundation under award \# 1855216.



\bibliographystyle{IEEEtran}
\bibliography{IEEEabrv,references}
%\begin{thebibliography}{1}
%
%%\bibitem{IEEEhowto:kopka}
%%H.~Kopka and P.~W. Daly, \emph{A Guide to \LaTeX}, 3rd~ed.\hskip 1em plus
%%  0.5em minus 0.4em\relax Harlow, England: Addison-Wesley, 1999.
%
%\end{thebibliography}

\end{document}